\documentclass[a4paper,12pt]{article}
\usepackage{amsfonts}
\usepackage{amsmath}
\usepackage{amssymb}
\usepackage{rotating}
\usepackage{graphicx}
\usepackage{setspace}
\usepackage[authoryear]{natbib}
\usepackage{tikz}
\usetikzlibrary{matrix,backgrounds}
\pgfdeclarelayer{myback}
\pgfsetlayers{myback,background,main}
\usetikzlibrary{topaths,positioning,chains,fit,shapes,calc}

\usepackage{enumitem}
\usepackage{booktabs}
\usepackage{multirow}
\usepackage{subcaption}
\usepackage{amsmath}

\usepackage{tabularx}
\usepackage{lscape}

\usepackage[colorlinks=true,linkcolor=blue]{hyperref}

\hypersetup{%
	citecolor=blue,
	breaklinks=true   
}

\newtheorem{theorem}{Theorem}

\newtheorem{lemma}{Lemma}

\newtheorem{proposition}[theorem]{Proposition}
\newtheorem{remark}{Remark}

\newenvironment{proof}[1][Proof]{\noindent\textbf{#1.} }{\ \rule{0.5em}{0.5em}}
\evensidemargin\oddsidemargin
\input{tcilatex}
\tikzstyle{vertex} = [fill,shape=circle,node distance=80pt]
\tikzstyle{edge} = [fill,opacity=.5,fill opacity=.5,line cap=round, line join=round, line width=50pt]
\tikzstyle{elabel} =  [fill,shape=circle,node distance=30pt]

\begin{document}

\title{Functional Differencing in Networks\thanks{%
		This paper has been prepared for the special issue of the Revue \'Economique in honor of Jean-Marc Robin. We thank the editor, an anonymous reviewer, Jaime Arellano-Bover, Jes\'us Carro, Martin Weidner, and Tom Zohar for comments.}}
\author{St\'ephane
	Bonhomme\thanks{%
		University of Chicago.} \and Kevin Dano\thanks{University of California Berkeley.}}

\date{$\quad $\\\today }
\vskip 3cm\maketitle

\begin{abstract}
	\noindent Economic interactions often occur in networks where heterogeneous agents (such as workers or firms) sort and produce. However, most existing estimation approaches either require the network to be dense, which is at odds with many empirical networks, or they require restricting the form of heterogeneity and the network formation process. We show how the functional differencing approach introduced by \citet{bonhomme2012functional} in the context of panel data, can be applied in network settings to derive moment restrictions on model parameters and average effects. Those restrictions are valid irrespective of the form of heterogeneity, and they hold in both dense and sparse networks. We illustrate the analysis with linear and nonlinear models of matched employer-employee data, in the spirit of the model introduced by \citet{abowd1999high}.  
	
	\bigskip

	
\end{abstract}

\baselineskip21pt

\bigskip

\bigskip

\setcounter{page}{0}\thispagestyle{empty}

\clearpage

\section{Introduction}

Network data is increasingly prevalent in applied economics. In this paper we focus on models where agents (e.g., workers and firms) sort and interact on a network. In such settings, accounting for unobserved heterogeneity is empirically key. However, existing approaches to estimation in the presence of flexible heterogeneity are imperfect.

A first approach consists in treating the heterogeneity as ``fixed effects'' parameters to be estimated. Bias reduction methods, initially developed for single-agent panel data (\citealp{hahn2004jackknife}, \citealp{dhaene2015split}, \citealp{fernandez2016individual}), have been recently extended to networks (e.g., \citealp{graham2017econometric}, \citealp{hughes2022estimating}). The fixed-effects approach is appealing since it does not require modeling the distribution of heterogeneity and how it correlates with conditioning variables. Additionally, in models where agents interact on an exogenous network, the fixed-effects approach does not require specifying a model of network formation. 

However, the performance of bias reduction methods in fixed-effects models hinges crucially on the network being sufficiently dense. This requirement is at odds with the nature of several empirical networks. For instance, in applications to wage determination in the presence of worker and firm heterogeneity (\citealp{abowd1999high}), a dense network approximation is typically inappropriate, and fixed-effects estimates suffer from a ``limited mobility bias'' that may be substantial (\citealp{bonhomme2020much}). More generally, sparsity is a feature of many empirical networks (\citealp{graham2020sparse}).

%
%

A second approach consists in postulating a ``random effects'' model for the unobserved heterogeneity. For example, \citet{bonhomme2019distributional} and \citet{lentz2022anatomy} propose and estimate random-effects models of worker heterogeneity in the presence of firm heterogeneity to account for sorting and complementarity on the labor market. Studying a different setting, \citet{bonhomme2021teams} develops random effects models of agent heterogeneity in team production networks. Random-effects methods enjoy theoretical guarantees in sparse networks, under the assumption that the model is correctly specified. 

However, modeling the full distribution of heterogeneity given conditioning variables can be challenging. In models of wage determination in the presence of worker and firm heterogeneity, this requires modeling the heterogeneity conditional on the entire network of employment relationships and job transitions, as proposed, for example, by \citet{woodcock2008wage} and \citet{bonhomme2020much}. More generally, the random-effects approach effectively requires modeling the network formation process, which can be a difficult task due to dimensionality and equilibrium multiplicity challenges.

In this paper our aim is to achieve the best of these two approaches, in the sense that we seek estimators that are fully robust to the form of heterogeneity, and that behave well in denser and sparser networks. Such estimators currently exist only in very special cases. Notably, \citet{andrews2008high} and \citet{kline2020leave} propose exact bias corrections for fixed-effects estimators of variance components in linear regressions on networks, while \citet{graham2017econometric} proposes a ``tetrad logit'' estimator in a logistic model of network formation. These strategies mimic panel data methods that are consistent in fixed-length panels, such as conditional logit estimators (\citealp{rasch1960studies}, \citealp{andersen1970asymptotic}) or estimators of variance components (\citealp{arellano2012identifying}).

In panel data, the functional differencing approach (\citealp{bonhomme2012functional}) provides a general methodology to find moment restrictions on parameters that are robust to any distribution of heterogeneity and correlation with conditioning variables, and hold in fixed-length panels. Recent applications of the approach include the derivation of new moment restrictions in binary and discrete choice models, both static and dynamic (\citealp{honore2020moment}, \citealp{honore2021dynamic}, \citealp{dano2023transition}). In certain models, no exact moment restrictions exist. However, \citet{dhaene2023approximate} show how to regularize the functional differencing moments to provide restrictions that are satisfied up to a vanishing approximation error as the length of the panel tends to infinity. 

The starting point of this paper is the observation that the scope of functional differencing is not limited to panel data, and the approach can be applied to any setting with a parametric conditional distribution involving latent variables. Our main goal is to apply the functional differencing approach to derive moment restrictions on parameters in some network settings. Specifically, we consider linear and binary choice logit models on networks, including a novel ``AKM logit model'' that provides a counterpart to the AKM estimator of \citet{abowd1999high} for binary outcomes. In those models, we characterize the available moment restrictions on parameters.

In addition, we study average effects that depend on the joint distribution of unobserved heterogeneity and observed covariates. In panel binary choice models, average effects have been studied by various authors (see, e.g., \citealp{chernozhukov2013average}, \citealp{davezies2021identification}, \citealp{dobronyi2021identification}, \citealp{aguirregabiria2021identification}, and \citealp{pakel2021bounds}). The functional differencing approach can be applied to average effects, as initially shown in the working paper version of \citet{bonhomme2012functional}. This approach applies to general panel data models where the outcome distribution is parametrically specified and the distribution of heterogeneity and covariates is unrestricted. Here we use it to derive moment conditions on average effects in network settings.

Lastly, as in its panel data applications, in the settings that we consider in this paper the functional differencing approach delivers moment restrictions on parameters, yet it does not guarantee identification or consistent estimation of the parameters given those restrictions. Although, in several examples that we study, identification can be verified directly and analog estimators can be constructed, applying the approach to other models will generally require careful analysis of statistical properties. We only briefly touch on estimation at the end of the paper, and leave a deeper analysis of identification and estimation to future work. At the same time, we see the functional differencing approach as a promising building block for researchers to discover novel moment restrictions and estimators in network settings in the future.

The outline of the paper is as follows. In Section \ref{sec_mod} we present a class of models with multi-sided heterogeneity in networks. In Section \ref{sec_FD} we describe the functional differencing approach. In Sections \ref{sec_ex_lin}, \ref{sec_ex_logit}, \ref{sec_ex_logitAKM}, and \ref{sec_logit_AME} we illustrate the approach with various examples. Finally, in Section \ref{sec_remarks} we briefly sketch how to construct estimators based on functional differencing moment restrictions.

\section{Models with multi-sided heterogeneity in networks\label{sec_mod}}

In this section we introduce and describe a framework for heterogeneous agents interacting on a network. The subsequent sections show how to derive moment restrictions on parameters and average effects in this framework.

\subsection{Description and examples}

 The model consists of two layers: a model of the network, and a model of agents' outcomes on the network.
 
The \emph{network} is represented by a graph, or more generally a hypergraph, featuring nodes and edges. Nodes in the network correspond to economic agents, and edges represent their links or collaborations. The network can be static or dynamic, and we will see that, under the assumption that the network is exogenous, a specification of the network formation model is not needed. 

A key feature of the framework is the presence of agent-specific \emph{types}, which govern sorting patterns and affect outcomes, yet are latent to the econometrician.  

Throughout the paper, we will refer to economic agents as workers and firms as a leading example. In settings with workers and firms, we model the network as bipartite, links are employment relationships, and both workers and firms are heterogeneous. We show an example in Figure \ref{Fig_bipart}. Employment, job mobility and wages all depend on the workers' and firms' latent types. Prominent models of this kind were proposed in \citet{becker1973theory}, \citet{shimer2000assortative}, \citet{postel2002equilibrium}, and \citet{lentz2022anatomy}, among many others.

\begin{figure}

\definecolor{mygrey}{RGB}{169,169,169}
\definecolor{myblack}{RGB}{0,0,0}

\caption{A worker-firm network\label{Fig_bipart}}	
\begin{center}
\begin{tikzpicture}[thick,
		every node/.style={circle},
		fsnode/.style={fill=mygrey},
		ssnode/.style={fill=myblack},
		]

		\begin{scope}[xshift=0cm,yshift=2cm,start chain=going below,node distance=10mm]
			\foreach \i in {1,2,...,5}
			\node[fsnode,on chain] (f\i) [label=left: \i] {};
		\end{scope}
		
		\begin{scope}[xshift=5cm,yshift=.5cm,start chain=going below,node distance=10mm]
			\foreach \i in {1,2,...,3}
			\node[ssnode,on chain] (s\i) [label=right: \i] {};
		\end{scope}
		
		\node [mygrey,fit=(f1) (f5),label=above:Workers] {};
		\node [myblack,fit=(s1) (s3),label=above:Firms] {};
		
		\draw (f1) -- (s1);
		\draw (f2) -- (s1);
		\draw (f2) -- (s2);
		\draw (f3) -- (s1);
		\draw (f3) -- (s3);
		\draw (f4) -- (s2);
		\draw (f5) -- (s3);
	\end{tikzpicture}
\end{center}
{\footnotesize\textit{Notes: Nodes are workers (numbered from 1 to 5) and firms (numbered from 1 to 3). Edges are lines linking workers to firms, representing employment relationships. Workers 1,4,5 stay in a single firm, while workers 2 and 3 move between two firms.}}
\end{figure}
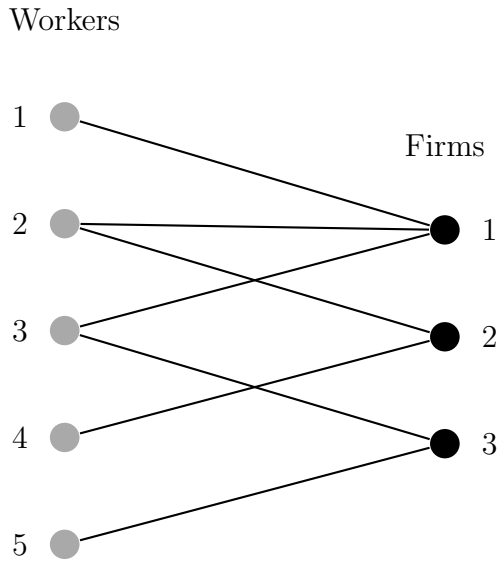

\emph{Outcomes} in the network depend on the agents' latent types, and on covariates that are observed by the econometrician. Outcomes also depend on idiosyncratic errors, or ``shocks''. 

A key assumption that we will maintain in this paper is that the network is \emph{exogenous}, in the sense that network links are independent of the shocks conditional on agents' types and covariates. In models with workers and firms, \citet{bonhomme2019distributional} show that this assumption is satisfied in the models of \citet{shimer2000assortative} and \citet{lentz2022anatomy} but that it fails in the model with sequential bargaining of \citet{postel2002equilibrium}, for example. 

Our main focus will be on estimating parameters governing the model of outcomes while conditioning on the network. This approach only restricts the network insofar as exogeneity with respect to idiosyncratic shocks is required. However, it allows for general forms of matching and sorting patterns. Beyond  the example of workers and firms, this setup is relevant to other bipartite network settings appearing in the economics of education, empirical finance, and trade (\citealp{bonhomme2020econometric}). Non-bipartite settings, such as models of team production, are also encompassed by this framework (\citealp{ahmadpoor2019decoding}, \citealp{bonhomme2021teams}).

In applications, the network formation model may also be of interest. Since the framework accounts for unobserved heterogeneity, it can be used to study network formation models with additive or non-additive heterogeneity and independent link-specific shocks (\citealp{graham2017econometric}, \citealp{bickel2009nonparametric}). However, since our approach relies on a parametric likelihood for the distribution of outcomes, it is not well-suited to study the determinants of link formation in models with strategic interactions (\citealp{de2018identifying}, \citealp{gualdani2021econometric}, \citealp{sheng2020structural}). 

\subsection{Probabilistic framework}

Let $Y$ denote a vector of outcomes, one for every link (or edge) in the network. Let $A$ denote the vector of latent types, one for every agent (or node). Finally, let $X$ denote a matrix indicating which agents are linked together. In models with covariates, which may be agent-specific or link-specific, we include those in $X$. 

We postulate a parametric model for outcomes conditional on the latent types and the network links (and possibly covariates),
\begin{equation}\label{eq_outcomes}
\left(Y\,|\, X=x,A=a\right) \sim f_{\theta}\left(\cdot\,|\, x,a\right),
\end{equation}
where $f_{\theta}$ is a parametric distribution indexed by a finite-dimensional parameter vector $\theta$.

We leave the joint distribution of latent types and network links (and possibly covariates) fully unrestricted. Formally,
\begin{equation}\label{eq_hetero}
	\left(A,X\right) \sim \pi,
\end{equation}
where $\pi$ is an unknown (i.e., nonparametric) distribution. 

The combination of a parametric outcome distribution and a nonparametric distribution of heterogeneity and covariates is common in the panel data literature. Indeed, (\ref{eq_outcomes})-(\ref{eq_hetero}) nests the standard panel data case, where $X$ simply indicates which observations correspond to the same worker. Nevertheless, the current framework is not limited to panel data. 

A key assumption implied by the specification of a likelihood conditional on $X$ (and $A$) is that the network is assumed exogenous. In panel data models, this corresponds to the assumption of strict exogeneity of covariates. This assumption is commonly relaxed in linear panel models, for example using the sequential moment restrictions introduced by \citet{arellano1991some}. A subsequent literature allows for sequentially exogenous network links in linear models (e.g., \citealp{kuersteiner2020dynamic}). Allowing for sequential exogeneity of covariates in nonlinear panel data models is still a frontier research area (see \citealp{bonhomme2023identification}). 

In a setting with workers and firms, a leading example of (\ref{eq_outcomes}) is a model of wage determination. Following the pioneering approach of \citet{abowd1999high}, a linear specification for log-wages is
\begin{equation}\label{eq_AKM}
Y_{it}=\alpha_i+\psi_{j(i,t)}+\varepsilon_{it},	
\end{equation}
where $i\in\{1,...,N\}$ are workers, $t\in\{1,...,T\}$ are time periods, $j(i,t)\in\{1,...,J\}$ denotes the firm where $i$ is employed at $t$, $\varepsilon_{it}$ is an idiosyncratic shock, and we are abstracting from exogenous covariates for simplicity. In this model, network exogeneity is often referred to as ``exogenous mobility'', meaning that job transitions (represented by the firm indicators $j(i,t)$) are assumed independent of the $\varepsilon_{it}$'s conditional on the $\alpha_i$'s and the $\psi_j$'s.

In order to allow for complementarity patterns, one may be interested in a nonlinear extension of (\ref{eq_AKM}), such as a Constant Elasticity of Substitution (CES) specification in logarithms,
\begin{equation}\label{eq_AKM_CES}
	Y_{it}=\ln \lambda+\frac{1}{\gamma}\ln\left( \rho\alpha_i^{\gamma}+(1-\rho)\psi_{j(i,t)}^{\gamma}\right)+\varepsilon_{it}.
\end{equation}
Nonlinear log wage specifications have been proposed and estimated by \citet{bonhomme2019distributional} and \citet{lentz2022anatomy}. However, a difference between their approaches and the one we propose here is that those authors restrict the distribution of heterogeneity and how it relates to employment relationships and job transitions. That is, in both papers, $\pi$ in (\ref{eq_hetero}) is restricted, while it is fully unrestricted in the current framework. In addition, both \citet{bonhomme2019distributional} and \citet{lentz2022anatomy} rely on large firms to recover firm heterogeneity, whereas our approach will yield valid moment restrictions irrespective of the degree of sparsity of the network. 

In models (\ref{eq_AKM}) and (\ref{eq_AKM_CES}) one can form an $NT\times (N+J)$  matrix $X$, by having each row denoting an observation (or edge) $(i,t)$, and having both sets of agents (or nodes) $i$ and $j$ in the columns of $X$. In this case, one can equivalently write (\ref{eq_AKM}) and (\ref{eq_AKM_CES}) as
$$(3')\,\,\, Y=XA+\varepsilon,\quad \text{ and }\quad (4')\,\,\, Y=g\left(X,A,\lambda,\rho,\gamma\right)+\varepsilon,$$
where the vector $A$ contains the $\alpha$'s and $\psi$'s, and the function $g$ takes a known (CES) form. Since we assume that the model is parametric, we specify the distribution of $\varepsilon$ conditional on $X$ and $A$, for example as a normal distribution with zero mean and diagonal covariance matrix with variance $\sigma^2$. In this case, $\theta=\sigma^2$ in the linear model, and $\theta=(\lambda,\rho,\gamma,\sigma^2)$ in the CES model.

The framework in (\ref{eq_outcomes})-(\ref{eq_hetero}) also nests certain models of network formation. As an example, consider the model of link formation in \citet{graham2017econometric}. Undirected binary links $Y_{ij}\in\{0,1\}$ between agents $i$ and $j$ are determined based on covariates $X_{ij}$ and agent-specific types $A_i$ and $A_j$ as
\begin{equation}\label{eq_graham}
	Y_{ij}=\boldsymbol{1}\left\{X_{ij}'\theta+A_i+A_j+\varepsilon_{ij}\geq 0\right\},
\end{equation}
where the $\varepsilon_{ij}$ are standard logistic, independent of $A$ and $X$, and i.i.d. across pairs of agents $(i,j)\in\{1,...,N\}^2$, $i\neq j$. \citet{graham2017econometric} leaves the joint distribution of $X$ and $A$ unrestricted, and his model is thus a special case of the framework in (\ref{eq_outcomes})-(\ref{eq_hetero}).

\section{Functional differencing: A general presentation\label{sec_FD}}

In the framework (\ref{eq_outcomes})-(\ref{eq_hetero}) we ask two questions. First, how can one derive moment restrictions on $\theta$? Second, how can one find moment restrictions on quantities depending on $\theta$ and $(A,X)$, such as average effects? Both questions can be answered using the functional differencing approach, and we address them in turn.

\subsection{Parameter $\theta$}

To answer the first question, the following proposition shows how to characterize moment restrictions on $\theta$ that are valid irrespective of the heterogeneity distribution.

\begin{proposition}\label{propo_FD}
	Let $\phi_{\theta}(y,x)$ be a function of outcomes and network links (and/or covariates). Suppose that $\mathbb{E}\left[\phi_{\theta}(Y,X)\,|\, A,X\right]$ is bounded. The following three statements are equivalent:
	
\noindent	(i) For any joint distribution of $(A,X)$ we have
	$$\mathbb{E}\left[\phi_{\theta}(Y,X)\right]=0.$$
	
\noindent	(ii) Almost surely in $A,X$, $$\mathbb{E}\left[\phi_{\theta}(Y,X)\,|\, A,X\right]=0.$$
	
\noindent	(iii) For almost all $a$ and $x$,
	\begin{equation}\int \phi_{\theta}(y,x)f_{\theta}(y\,|\, x,a)dy=0.\label{eq_FD}\end{equation}
	\end{proposition}

It is immediate that (ii) implies (i), and that (ii) and (iii) are equivalent. The implication from (i) to (ii) comes from the fact that we require the moment restriction to hold irrespective of the distribution of $A$ and $X$. A formal argument is provided in Appendix \ref{app_proofs}. Note also that (ii) implies moment restrictions conditional on $X$, $\mathbb{E}\left[\phi_{\theta}(Y,X)\,|\, X\right]=0$. Proposition \ref{propo_FD} implies that, to look for moment restrictions on $\theta$, it is necessary and sufficient to find solutions to the linear functional equation in (\ref{eq_FD}).    

Proposition \ref{propo_FD} is the main insight of the functional differencing approach (\citealp{bonhomme2012functional}). Indeed, finding a $\phi$ satisfying (\ref{eq_FD}) amounts to finding an element in the null space of the conditional expectation operator associated with the parametric conditional model of $Y$ given $(A,X)$. To this end, \citet{bonhomme2012functional} proposed numerical projection methods, while \citet{honore2020moment} relied on symbolic computing to find analytical $\phi$ functions in dynamic discrete choice settings. In some models, however, one can show that no non-trivial solution $\phi$ exists, which implies the absence of informative moment equality restrictions on $\theta$. To deal with such cases, \citet{dhaene2023approximate} propose an approximate functional differencing approach. In dynamic panel data logit models, \citet{dobronyi2021identification} derive the identified set on the parameters, which combines moment equality and inequality restrictions.

While most applications of the functional differencing approach so far are confined to panel data settings without interactions between agents, Proposition \ref{propo_FD} makes clear that the scope of the approach is not limited to those settings. A conventional panel data model consists of a collection of individual-specific submodels indexed by the individual heterogeneity $A_i$, the vector of individual covariates $(X_{i1}',...,X_{iT}')'$, and the parameter $\theta$ that is common across individuals. In contrast, in network settings all units may potentially be related, and both the network matrix $X$ and the vector of heterogeneity $A$ may affect all outcomes in the network. As Proposition \ref{propo_FD} illustrates, this difference between panel data and network settings is immaterial from the point of view of the applicability of functional differencing. In the next sections we will provide examples to illustrate the usefulness of the functional differencing approach when applied to networks.   

\subsection{Average effects}

We now turn our attention to linear functionals of the distribution of $A$ and $X$, and answer the second question. Using functional differencing, we show how to obtain, for given $\theta$, moment restrictions on an average effect of the form 
$$\mu=\mathbb{E}[m_{\theta}(A,X)],$$
where $m_{\theta}(\cdot)$ is known given $\theta$, and the expectation is taken with respect to the joint distribution $\pi$ of $(A,X)$. As an example, in the CES model of wage determination (\ref{eq_AKM_CES}), one may be interested in estimating average marginal effects of worker or firm heterogeneity on log wages, while accounting for the presence of complementarity between worker and firm effects. The following proposition provides a counterpart to Proposition \ref{propo_FD} for such target parameters.

\begin{proposition}\label{prop_FD_marg}

	Let $\psi_{\theta}(y,x)$ be a function of outcomes and network links (and/or covariates). Suppose that $\mathbb{E}\left[\psi_{\theta}(Y,X)\,|\, A,X\right]-m_{\theta}(A,X)$ is bounded. The following three statements are equivalent:

\noindent (i) For any joint distribution of $(A,X)$ we have
$$\mathbb{E}\left[\psi_{\theta}(Y,X)\right]=\mathbb{E}[m_{\theta}(A,X)].$$

\noindent (ii) Almost surely in $A,X$, $$\mathbb{E}\left[\psi_{\theta}(Y,X)\,|\, A,X\right]=m_{\theta}(A,X).$$

\noindent (iii) For almost all $a$ and $x$,
\begin{equation}\int \psi_{\theta}(y,x)f_{\theta}(y\,|\, x,a)dy=m_{\theta}(a,x).\label{eq_FD_marg}\end{equation}

\end{proposition}

The equivalence between the three parts is again easy to see (see Appendix \ref{app_proofs}), and the usefulness of the result comes from the fact that (\ref{eq_FD_marg}) is a linear functional equation.  In this case as well, the linear operator in (\ref{eq_FD_marg}) is known given $\theta$. Proposition \ref{prop_FD_marg} shows that the functional differencing approach initially applied in \citet{bonhomme2012functional} to obtain restrictions on model parameters in panel data models can be applied to derive restrictions on average effects in other models with latent variables, including network settings.

Finding a moment representation for $\mu$ amounts to solving a linear functional system, which is a Fredholm integral equation of the first kind. Numerical and analytical methods can be used to construct $\psi$ functions that satisfy (\ref{eq_FD_marg}). In dynamic panel logit models,  \citet{aguirregabiria2021identification}, \citet{dobronyi2021identification} and \citet{dano2023transition} show that the computation of average marginal effects is analytically straightforward. However, in models with continuous outcomes the inverse problem in (\ref{eq_FD_marg}) is generally ill-posed (\citealp{carrasco2007linear}, \citealp{engl1996regularization}). Given a solution $\psi$ to the functional system, using it for estimation in a finite sample thus typically requires regularization. In this paper we focus on finding moment functions $\phi$ and $\psi$, and leave a detailed study of estimators and their properties to future work. See Section \ref{sec_remarks} for further discussion.


\section{Linear network models\label{sec_ex_lin}}

In this section and the next three we illustrate Propositions \ref{propo_FD} and \ref{prop_FD_marg} through various examples. Consider first a linear model with an $X$ matrix that consists of two parts, $X=(X_1,X_2)$, where $X_1$ is a matrix of network links and $X_2$ is a matrix of covariates. An example is the AKM model of \citet{abowd1999high}, given by an augmented version of (\ref{eq_AKM}) that includes covariates. We specify
\begin{equation}Y=X_1A+X_2\beta+\varepsilon,\quad (\varepsilon\,|\, X,A)\sim iid{\cal{N}}(0,\sigma^2 I_n),\label{eq_mod_lin}\end{equation}
where $n$ is the number of observations, $I_n$ denotes the $n\times n$ identity matrix, and we denote $\theta=(\beta',\sigma^2)'$.

In this model we will focus on the parameters $\beta$ and $\sigma^2$, and on quadratic forms $\mu=\mathbb{E}[A'QA]$ for a symmetric matrix $Q$. Variance components, which can be written as quadratic forms, are of interest for, e.g., decomposing the variance of log wages into components reflecting worker heterogeneity, firm heterogeneity, and sorting patterns between heterogeneous workers and firms (\citealp{abowd1999high}, \citealp{card2013workplace}, \citealp{song2019firming}). 

\subsection{Parameters $\beta$ and $\sigma^2$}

In this subsection we derive moment restrictions on $\beta$ and $\sigma^2$. For this purpose we rely on Proposition \ref{propo_FD}. We start by noting that (\ref{eq_FD}) can be equivalently written as
\begin{align*}
&\int \phi_{\theta}(y,x)\exp\left(-\frac{1}{2\sigma^2}(y-x_1a-x_2\beta)'(y-x_1a-x_2\beta)\right)dy=0.\label{eq_lin1}
\end{align*}

It is useful to introduce the Moore-Penrose pseudo-inverse $x_1^{\dagger}$  of $x_1$, and to write the two orthogonal projectors associated with $x_1$ as $x_1x_1^{\dagger}=u_1u_1'$ and $I_n-x_1x_1^{\dagger}=u_2u_2'$, where $u=(u_1,u_2)$ is orthogonal. Let $v=y-x_2\beta$, $v_1=u_1'v$, and $v_2=u_2'v$. 

We then note that (\ref{eq_FD}) can equivalently be written as
\begin{align*}
&\int \phi_{\theta}(y,x)\exp\left(-\frac{1}{2\sigma^2}(y-x_1a-x_2\beta)'x_1x_1^{\dagger}(y-x_1a-x_2\beta)\right)\\
	&\quad\quad\quad\quad\quad\quad\times\exp\left(-\frac{1}{2\sigma^2}(y-x_2\beta)'[I_n-x_1x_1^{\dagger}](y-x_2\beta)\right)dy=0.	\end{align*}
Since $u=(u_1,u_2)$ is orthogonal, we equivalently obtain 
\begin{align*}
&  \int \bigg[\int  \phi_{\theta}(x_2\beta+u_1v_1+u_2v_2,x)\exp\left(-\frac{1}{2\sigma^2}v_2'v_2\right)dv_2\bigg]\\
&\quad\quad\quad\quad\quad\quad\times \exp\left(-\frac{1}{2\sigma^2}(v_1-u_1'x_1a)'(v_1-u_1'x_1a)\right)dv_1=0. 	\end{align*}
In Proposition \ref{propo_FD} we are looking for restrictions holding for all real vectors $a$. Since the rows of $u_1'x_1$ are linearly independent, we equivalently search for the following equation being satisfied for all real vectors $b$:\footnote{$b$ has the same dimension as $v_1$.}
\begin{align*}
	&  \int \bigg[\int  \phi_{\theta}(x_2\beta+u_1v_1+u_2v_2,x)\exp\left(-\frac{1}{2\sigma^2}v_2'v_2\right)dv_2\bigg]\\
	&\quad\quad\quad\quad\quad\quad\times \exp\left(-\frac{1}{2\sigma^2}(v_1-b)'(v_1-b)\right)dv_1=0. 	\end{align*}
This convolution equation has the unique solution:
\begin{align*}
	&  \int \phi_{\theta}(x_2\beta+u_1v_1+u_2v_2,x)\exp\left(-\frac{1}{2\sigma^2}v_2'v_2\right)dv_2=0.	\end{align*}

We have thus shown the following.

\begin{proposition}\label{prop_linear_param}
	In model (\ref{eq_mod_lin}), the following two statements are equivalent:
	
\noindent	(i) $\int \phi_{\theta}(y,x)f_{\theta}(y\,|\, x,a)dy=0$.
	
\noindent	(ii) $\phi_{\theta}(y,x)=\varphi_{\theta}(u_1'(y-x_2\beta),u_2'(y-x_2\beta),x)$, where the function $\varphi_{\theta}$ is such that
	\begin{align*}
		&\int \varphi_{\theta}(v_1,v_2,x)\exp\left(-\frac{1}{2\sigma^2}v_2'v_2\right)dv_2=0.	\end{align*}
	
\end{proposition}

By Proposition \ref{propo_FD}, Proposition \ref{prop_linear_param} characterizes all the available moment restrictions on the parameter $\theta=(\beta',\sigma^2)$ in model (\ref{eq_mod_lin}). Now, there are many possible choices for $\varphi_{\theta}$, leading to many choices of moment functions in this model. 

As a first example, let us take
$$\varphi_{\theta}(v_1,v_2,x)=u_2v_2.$$ 
We obtain
\begin{eqnarray*}
\phi_{\theta}(y,x)&=&\varphi_{\theta}(u_1'(y-x_2\beta),u_2'(y-x_2\beta),x)\notag\\
&=&u_2u_2'(y-x_2\beta)\\
&=&(I_n-x_1x_1^{\dagger})(y-x_2\beta).
\end{eqnarray*}
This implies the following conditional moment restrictions on $\beta$:
\begin{equation}
\mathbb{E}\left[(I_n-X_1X_1^{\dagger})(Y-X_2\beta)\,|\, X_1,X_2\right]=0.\label{eq_restr_beta}
\end{equation}

In a panel data setting, \citet{chamberlain1992efficiency} shows that the efficiency bound for $\beta$ based on the quasi-differencing restrictions (\ref{eq_restr_beta}) coincides with the bound based on a semiparametric model where $\varepsilon$ has mean zero but is otherwise not restricted. In our setup, (\ref{eq_restr_beta}) remains valid in general non-Gaussian linear network regression models where $\mathbb{E}[\varepsilon\,|\, X]=0$. Moreover, the Gaussian assumption on $\varepsilon$ provides additional restrictions that are fully characterized by Proposition \ref{prop_linear_param}. Some of those additional restrictions arise from the variance matrix of outcomes.

As a second example, let $n_2=\dim v_2=\mbox{Trace}(I_n-x_1x_1^{\dagger})$, and take $$\varphi_{\theta}(v_1,v_2,x)=v_2'v_2-n_2\sigma^2 .$$ 
We obtain $$\phi_{\theta}(y,x)=(y-x_2\beta)'[I_n-x_1x_1^{\dagger}](y-x_2\beta)-n_2\sigma^2.$$
This implies the following conditional moment restrictions on $\beta$ and $\sigma^2$:
\begin{equation}
	\mathbb{E}\left[(Y-X_2\beta)'[I_n-X_1X_1^{\dagger}](Y-X_2\beta)-n_2\sigma^2\,|\, X_1,X_2\right]=0.\label{eq_restr_sigma2}
\end{equation}
Note that (\ref{eq_restr_sigma2}) remains valid under non-Gaussianity, provided $\mathbb{E}[\varepsilon\varepsilon'\,|\, X]=\sigma^2 I_n$. Restrictions akin to (\ref{eq_restr_sigma2}) were considered in \citet{arellano2012identifying} in a panel data setting. In a network context, \citet{andrews2008high} derived an unconditional version of (\ref{eq_restr_sigma2}), and applied it to the decomposition of the variance of log wages.\footnote{The assumption that the elements of $\varepsilon$ be mutually independent may be empirically restrictive. In applications of AKM, it is common to only rely on between-job-spell variation in log wages in estimation, in order not to restrict the within-job-spell correlation in $\varepsilon$ (\citealp{kline2020leave}, \citealp{bonhomme2020much}).}

\subsection{Quadratic forms}

In this subsection we derive moment restrictions on a quadratic form $\mu=\mathbb{E}[m_{\theta}(A,X)]$, where $m_{\theta}(a,x)=a'Qa$ for an $m\times m$ symmetric matrix $Q$, for $m$ the dimension of $A$. For this purpose we rely on Proposition \ref{prop_FD_marg}. We start by noting that (\ref{eq_FD_marg}) is equivalent to
\begin{align*}
\int \psi_{\theta}(y,x)\frac{1}{(2\pi \sigma^2)^{\frac{n}{2}}}\exp\left(-\frac{1}{2\sigma^2}(y-x_1a-x_2\beta)'(y-x_1a-x_2\beta)\right)dy=a'Qa.
\end{align*}
Using the above reparameterization in terms of $(v_1,v_2)$, we equivalently have 
\begin{align}
&\frac{1}{(2\pi \sigma^2)^{\frac{n}{2}}}\iint \psi_{\theta}(x_2\beta+u_1v_1+u_2v_2,x)\notag\\
	&\quad \times\exp\left(-\frac{1}{2\sigma^2}(v_1-u_1'x_1a)'(v_1-u_1'x_1a)\right)\exp\left(-\frac{1}{2\sigma^2}v_2'v_2\right)dv_1dv_2=a'Qa.\label{eq_prop4_forproof}	\end{align}

We then have the following result, shown in Appendix \ref{app_proofs}.

\begin{proposition}\label{prop_linear_quad}
	Consider model (\ref{eq_mod_lin}), and suppose that $X_1$ has full column rank almost surely. Then the following two statements are equivalent:
	
\noindent	(i) $\int \psi_{\theta}(y,x)f_{\theta}(y\,|\, x,a)dy=a'Qa$.
	
\noindent	(ii) $\frac{1}{(2\pi \sigma^2)^{\frac{n_2}{2}}}\int \psi_{\theta}(x_2\beta+u_1v_1+u_2v_2,x)\exp\left(-\frac{1}{2\sigma^2}v_2'v_2\right)dv_2=v'(x_1^{\dagger})'Qx_1^{\dagger}v-\sigma^2\mbox{Trace}((x_1^{\dagger})'Qx_1^{\dagger})$.
	
\end{proposition}

By Proposition \ref{prop_FD_marg}, Proposition \ref{prop_linear_quad}
 characterizes all the available moment restrictions on $\mu=\mathbb{E}[A'QA]$.\footnote{When $X_1$ is a network matrix, ensuring the assumption that $X_1$ has full column rank often requires to restrict the sample to a connected subnetwork. See \citet{abowd1999high} and \citet{abowd2002computing} for methods to compute connected subnetworks in settings with workers and firms.} The proof relies on Fourier transforms. A special case of Proposition \ref{prop_linear_quad} is obtained when $\psi_{\theta}(y,x)$ is a function of $u_1'(y-x_2\beta)$ and $x$ only, which implies
\begin{equation}\label{eq_psi_lin}\psi_{\theta}(y,x)=(y-x_2\beta)'(x_1^{\dagger})'Qx_1^{\dagger}(y-x_2\beta)-\sigma^2\mbox{Trace}((x_1^{\dagger})'Qx_1^{\dagger}).\end{equation}
The trace correction in (\ref{eq_psi_lin}) is a well-known formula to obtain unbiased estimators of quadratic forms. \citet{andrews2008high}, and \citet{kline2020leave} in a heteroskedastic context, apply such corrections to estimate variance components in log wage variance decompositions.

Moreover, given the particular solution (\ref{eq_psi_lin}), any other solution is of the form
$$\widetilde{\psi}_{\theta}(y,x)=\psi_{\theta}(y,x)+\varphi_{\theta}(u_1'(y-x_2\beta),u_2'(y-x_2\beta),x),$$
where, as in Proposition \ref{prop_linear_param}, the function $\varphi_{\theta}$ satisfies\begin{align*}
	&\int \varphi_{\theta}(v_1,v_2,x)\exp\left(-\frac{1}{2\sigma^2}v_2'v_2\right)dv_2=0.	\end{align*}

\begin{remark}

Proposition \ref{prop_FD_marg} can be applied to find estimators of other average effects beyond quadratic forms. Consider the quantity $\mu=\mathbb{E}[m_{\theta}(A,X)]$, for some function $m_{\theta}(\cdot)$. Examples are higher-order moments of $A$, its distribution function, or some nonlinear moments of $A$. Using similar arguments to the ones leading to Proposition \ref{prop_linear_quad}, one can derive a formula for all moment restrictions on $\mu$, 
\begin{equation}\mathbb{E}[\psi_{\theta}(Y,X)]=\mu.\label{eq_mom_mu}\end{equation}
However, the expression of $\psi_{\theta}$, which we derive in Appendix \ref{app_linear_nonquad}, depends on the inverse of the Fourier transform operator, and it does not generally admit a closed-form expression when $m_{\theta}(a,x)$ is not polynomial in $a$. In addition, estimating $\mu$ based on (\ref{eq_mom_mu}) generally requires regularizing $\psi_{\theta}$ (as shown in a panel data setting in the working paper version of \citealp{bonhomme2012functional}). 


\end{remark}

%
%
%
%
%
%
%

\section{Logit network models\label{sec_ex_logit}}

In this section and the next two we study logit models for network data. Letting $X=(X_1,X_2)$, we assume
\begin{equation}Y=\boldsymbol{1}\left\{X_1A+X_2\theta+\varepsilon>0\right\},\quad \left(\varepsilon\,|\, X,A\right)\sim iidLogistic.\label{eq_mod_logit}\end{equation}

Model (\ref{eq_mod_logit}) contains several popular binary choice models as special cases. A first example is a static panel data logit model, which obtains when the columns of $X_1$ are individual indicators. Another example is the logistic network formation model of \citet{graham2017econometric}, see (\ref{eq_graham}), which obtains when $Y$ are link outcomes and the elements of $X_1 A$ are the individual sums $A_i+A_j$, for $i\neq j$.   

Equation (\ref{eq_mod_logit}) also covers models of binary outcomes on a network. To illustrate, we will consider the following binary choice counterpart to the AKM model of \citet{abowd1999high}:
\begin{equation}\label{eq_AKM_logit}
	Y_{it}=\boldsymbol{1}\left\{X_{it}'\theta+\alpha_i+\psi_{j(i,t)}+\varepsilon_{it}>0\right\},\quad  (\varepsilon_{it}\,|\, X_{11},...,X_{NT})\sim iidLogistic,	
\end{equation}
where, as in the linear case, $i$ are workers, $t$ are time periods, and $j(i,t)$ denotes the firm where $i$ is employed at $t$. It is easy to see that (\ref{eq_AKM_logit}) can be written in the form (\ref{eq_mod_logit}) for a suitable definition of $X_1$. In this setting, the network is a bipartite multigraph: there may be multiple edges pointing from a worker $i$ to a firm $j$ indicating that worker $i$ was in an employment relationship with firm $j$ over multiple periods.

While, in many applications of AKM, outcomes $Y_{it}$ are log earnings or log wages, it may be of interest to account for worker and firm heterogeneity when studying other labor market outcomes. For example, \citet{lachowska2023work} apply AKM to the analysis of log hours worked. In this setting, the AKM logit specification (\ref{eq_AKM_logit}) could be employed to analyze the determinants of part-time and full-time work using a binary measure of working time. In other applications, one may be interested in applying AKM to study determinants of worker promotions (e.g., \citealp{benson2019promotions}) or of the type of labor contract such as fixed-term or permanent contract (e.g., \citealp{guell2007binding}). The logit specification (\ref{eq_AKM_logit}) can also be useful in applications of AKM to other fields (including education, innovation, urban economics, trade, and empirical finance) where binary outcomes are common.        

In the next section we will first focus on the covariates' coefficients $\theta$ in model (\ref{eq_mod_logit}). For example, \citet{margolis1996cohort} studies the earnings returns to seniority in France while accounting for worker and firm heterogeneity. A binary specification such as (\ref{eq_AKM_logit}) allows one to document the effects of seniority on binary labor market outcomes, such as working part-time or full-time, being awarded a promotion, or working under a permanent or temporary contract. In Section \ref{sec_logit_AME} we will study average effects, which are functions of worker and firm heterogeneity. We will see that deriving non-trivial moment restrictions on average effects seems more challenging than obtaining informative restrictions about the $\theta$ parameter.

\section{Model parameters in logit network models\label{sec_ex_logitAKM}}

In this section we first provide a characterization of all moment restrictions available on $\theta$ in model (\ref{eq_mod_logit}), and then discuss several examples.

\subsection{Characterization}

We start by noting that, in model (\ref{eq_mod_logit}), (\ref{eq_FD}) can be equivalently written as 
\begin{align*}
	\sum_{y\in\{0,1\}^n}\phi_{\theta}(y,x)\prod_{i=1}^n\left(\frac{\exp(x_{i1}'a+x_{i2}'\theta)}{\exp(x_{i1}'a+x_{i2}'\theta)+1}\right)^{y_i} \left(\frac{1}{\exp(x_{i1}'a+x_{i2}'\theta)+1}\right)^{1-y_i}=0,
\end{align*}	
where we have denoted as $y_i$ the $i$th element of $y$, $x_{i1}'$ the $i$th row of $x_1$, and $x_{i2}'$ the $i$th row of $x_2$, for $i\in\{1,...,n\}$.

Equivalently, we have
\begin{align*}	
	& \sum_{y\in\{0,1\}^n}\phi_{\theta}(y,x)\prod_{i=1}^n\exp(y_ix_{i1}'a+y_ix_{i2}'\theta)=0,
\end{align*}	
that is, 
\begin{align*}	
	&\sum_{y\in\{0,1\}^n}\phi_{\theta}(y,x)\exp\left(y'x_2\theta\right)\exp\left(y'x_1a\right)=0.
\end{align*}

Letting, for all $x_1$, ${\cal{S}}_{x_1}=\{x_1'y\, : \, y\in\{0,1\}^n\}$ denote the set of possible values of $x_1'y$, this implies that (\ref{eq_FD}) can be equivalently written as 
\begin{align*}	
	& \sum_{s\in{\cal{S}}_{x_1}}\sum_{y\in\{0,1\}^n}\boldsymbol{1}\left\{x_1'y=s\right\}\phi_{\theta}(y,x)\exp\left(y'x_2\theta\right)\exp\left(s'a\right)=0.\end{align*}

Finally, since $\exp(s'a)$, for $s\in{\cal{S}}_{x_1}$, are linearly independent functions of $a$, we obtain the following characterization.

%

\begin{proposition}\label{prop_logit}
	In model (\ref{eq_mod_logit}), the following two statements are equivalent:
	
	\noindent	(i) $\int \phi_{\theta}(y,x)f_{\theta}(y\,|\, x,a)dy=0$.
	
	\noindent	(ii) $\sum_{y\in\{0,1\}^n}\boldsymbol{1}\left\{x_1'y=s\right\}\phi_{\theta}(y,x)\exp\left(y'x_2\theta\right)=0$, for all $s\in{\cal{S}}_{x_1}$.
\end{proposition} 

Proposition \ref{prop_logit} provides an exhaustive characterization of available moment restrictions in the logit model (\ref{eq_mod_logit}). For a non-zero $\phi$ to exist, it is necessary that, for some $s\in{\cal{S}}_{x_1}$, $y'x_2$ varies given that $x_1'y=s$. In this model, $S=X_1'Y$ is sufficient for $A$. Below we will illustrate Proposition \ref{prop_logit} using several examples: the static panel data logit model, the logistic network formation model, and the AKM logit model. 

%
%


%

\subsection{Panel data: conditional logit} 

Consider first the panel data model
\begin{equation} \label{dgp_static_logit}
	Y_{it}=\boldsymbol{1}\left\{A_i+X_{it}'\theta+\varepsilon_{it}>0\right\}, \quad (\varepsilon_{it}\,|\, A,X)\sim iidLogistic,\quad i=1,...,N,\,\,t=1,...,T.
\end{equation}
Let $i\in\{1,...,N\}$, $y_i=(y_{i1},...,y_{iT})'$ and $x_{i}=(x_{i1}',...,x_{iT}')'$. For simplicity we search for functions $\phi_{\theta}(y_i,x_{i})$ that only depend on $(y,x)$ through $(y_i,x_i)$. We thus look for $\phi_{\theta}(y_i,x_{i})$ such that
\begin{equation}
	\sum_{y_i\in\{0,1\}^T}\boldsymbol{1}\left\{\sum_{t=1}^Ty_{it}=s\right\}\phi_{\theta}(y_i,x_{i})\exp\left(\sum_{t=1}^Ty_{it}x_{it}'\theta\right)=0,\quad s=0,1,...,T.\label{eq_phi_panel}
\end{equation}

Consider the $T=2$ case, and take $s=1$. We obtain
$$\phi_{\theta}(1,0,x_i)\exp\left(x_{i1}'\theta\right)+\phi_{\theta}(0,1,x_{i})\exp\left(x_{i2}'\theta\right)=0,$$
which coincides with the moment function of conditional logit (\citealp{rasch1960studies}, \citealp{andersen1970asymptotic}). When $T>2$ we recover additional moment restrictions, as in \citet{davezies2020fixed}.

%
%

\subsection{Network formation: tetrad logit}

Consider next the logistic network formation model (\ref{eq_graham}) introduced in \citet{graham2017econometric}. Links are undirected,\footnote{Logit models with directed links (\citealp{charbonneau2017multiple}) have a similar structure.} and there are $n=N(N-1)/2$ observations (one for each  dyad), where $N$ is the number of agents. In this model,  the sufficient statistic $s=x_1'y$ in Proposition \ref{prop_logit} is the vector of degrees, i.e., the degree sequence of the network.

For simplicity we focus on functions of tetrads formed by four agents $(i,j,k,\ell)$. We thus look for $\phi_{\theta}(y_{ij},y_{ik},y_{i\ell},y_{jk},y_{j\ell},y_{k\ell},x)$
that satisfy
\begin{align*}
	& \sum_{{\tiny\begin{array}{c}y_{ij},y_{ik},y_{i\ell}\\y_{jk},y_{j\ell},y_{k\ell}\end{array}}}\boldsymbol{1}\left\{y_{ij}+y_{ik}+y_{i\ell}=s_1,y_{ij}+y_{jk}+y_{j\ell}=s_2,y_{jk}+y_{ik}+y_{k\ell}=s_3,y_{i\ell}+y_{j\ell}+y_{k\ell}=s_4\right\}\\
	&\times\phi_{\theta}(y_{ij},y_{ik},y_{i\ell},y_{jk},y_{j\ell},y_{k\ell},x)\exp\left(\left(y_{ij}x_{ij}+y_{ik}x_{ik}+y_{i\ell}x_{i\ell}+y_{jk}x_{jk}+y_{j\ell}x_{j\ell}+y_{k\ell}x_{k\ell}\right)'\theta\right)=0. \end{align*}

Taking first $s_1=s_2=s_3=s_4=1$, we obtain  
\begin{align}
	& \phi_{\theta}(1,0,0,0,0,1,x)\exp\left(\left(x_{ij}+x_{k\ell}\right)'\theta\right)+\phi_{\theta}(0,1,0,0,1,0,x)\exp\left(\left(x_{ik}+x_{j\ell}\right)'\theta\right)\notag\\
	&\quad\quad \quad+\phi_{\theta}(0,0,1,1,0,0,x)\exp\left(\left(x_{i\ell}+x_{jk}\right)'\theta\right)=0.\label{eq_homoph1} \end{align}
Considering next $s_1=s_2=s_3=s_4=2$,  we obtain  
\begin{align}
	& \phi_{\theta}(1,1,0,0,1,1,x)\exp\left(\left(x_{ij}+x_{ik}+x_{j\ell}+x_{k\ell}\right)'\theta\right)\notag\\
	&\quad+\phi_{\theta}(1,0,1,1,0,1,x)\exp\left(\left(x_{ij}+x_{i\ell}+x_{jk}+x_{k\ell}\right)'\theta\right)\notag\\
	&\quad\quad+\phi_{\theta}(0,1,1,1,1,0,x)\exp\left(\left(x_{ik}+x_{i\ell}+x_{jk}+x_{j\ell}\right)'\theta\right)=0.\label{eq_homoph2} \end{align}
Finally, for $s_1=s_2=2,s_3=s_4=1$, we obtain
\begin{align}
	& \phi_{\theta}(1,1,0,0,1,0,x)\exp\left(\left(x_{ij}+x_{ik}+x_{j\ell}\right)'\theta\right)\notag\\
	&\quad\quad \quad+\phi_{\theta}(1,0,1,1,0,0,x)\exp\left(\left(x_{ij}+x_{i\ell}+x_{jk}\right)'\theta\right)=0,\label{eq_homoph3} \end{align}
and there will be analogous restrictions associated with permutations of the degree sequence (2,2,1,1). All other possible degree sequences have no identifying content for $\theta$. 

Together, taking $\phi$ as in (\ref{eq_homoph1}), (\ref{eq_homoph2}), and (\ref{eq_homoph3}) (alongside its permutations), implies the moment restrictions underpinning the ``tetrad logit'' estimator in \citet{graham2017econometric}. However, Proposition \ref{prop_logit} clarifies that these restrictions may not be unique, and it provides all available moment restrictions in the logistic network formation model. 

\subsection{Binary choice on a network: AKM logit}

In this subsection we derive moment restrictions on $\theta$ in model (\ref{eq_AKM_logit}). We focus on the case where $X_{it}$ does not vary within job spells. For example, when controlling for the worker's age, job seniority only varies between spells.\footnote{If $X_{it}$ does vary within spells, then the conditional logit estimator can be used for consistent estimation of $\theta$.} We focus the analysis on the $T=2$ case, and we consider several subnetwork configurations of the data displayed in Figure \ref{Fig_bipart2}. Given a subnetwork configuration, we verify if moment conditions on $\theta$ exist, and what form they take. 


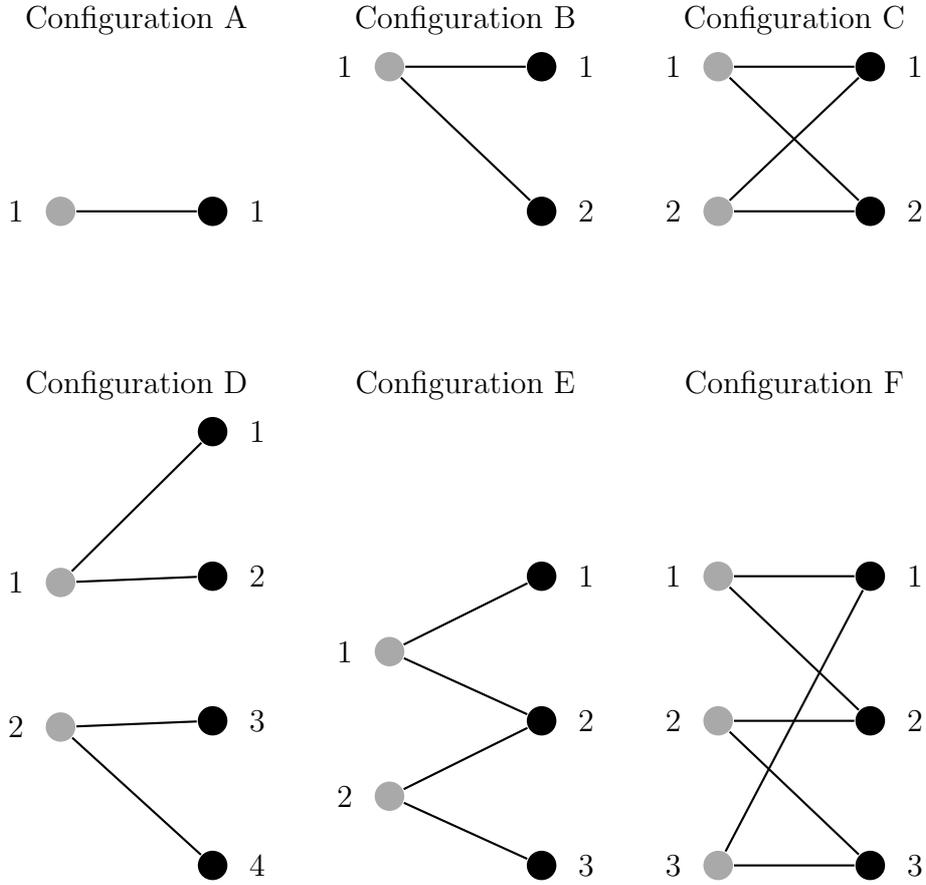
\begin{figure}[tbp]
	
	\definecolor{mygrey}{RGB}{169,169,169}
	\definecolor{myblack}{RGB}{0,0,0}

	\caption{Configurations of worker-firm subnetworks in the logit model with worker and firm heterogeneity\label{Fig_bipart2}}	
	\begin{center}
		\begin{tabular}{ccc}
			Configuration A & Configuration B & Configuration C\\	
			\begin{tikzpicture}[thick,
				every node/.style={circle},
				fsnode/.style={fill=mygrey},
				ssnode/.style={fill=myblack},
				]

				\begin{scope}[xshift=0cm,yshift=0cm,start chain=going below,node distance=15mm]
					\foreach \i in {1}
					\node[fsnode,on chain] (f\i) [label=left: \i] {};
				\end{scope}
				
				\begin{scope}[xshift=2cm,yshift=0cm,start chain=going below,node distance=15mm]
					\foreach \i in {1}
					\node[ssnode,on chain] (s\i) [label=right: \i] {};
				\end{scope}
				
				\draw (f1) -- (s1);
			\end{tikzpicture}
			&\begin{tikzpicture}[thick,
				every node/.style={circle},
				fsnode/.style={fill=mygrey},
				ssnode/.style={fill=myblack},
				]

				\begin{scope}[xshift=0cm,yshift=0cm,start chain=going below,node distance=15mm]
					\foreach \i in {1}
					\node[fsnode,on chain] (f\i) [label=left: \i] {};
				\end{scope}
				
				\begin{scope}[xshift=2cm,yshift=0cm,start chain=going below,node distance=15mm]
					\foreach \i in {1,2}
					\node[ssnode,on chain] (s\i) [label=right: \i] {};
				\end{scope}
				
				\draw (f1) -- (s1);
				\draw (f1) -- (s2);
			\end{tikzpicture}
			&\begin{tikzpicture}[thick,
				every node/.style={circle},
				fsnode/.style={fill=mygrey},
				ssnode/.style={fill=myblack},
				]

				\begin{scope}[xshift=0cm,yshift=0cm,start chain=going below,node distance=15mm]
					\foreach \i in {1,2}
					\node[fsnode,on chain] (f\i) [label=left: \i] {};
				\end{scope}
				
				\begin{scope}[xshift=2cm,yshift=0cm,start chain=going below,node distance=15mm]
					\foreach \i in {1,2}
					\node[ssnode,on chain] (s\i) [label=right: \i] {};
				\end{scope}
				
				\draw (f1) -- (s1);
				\draw (f1) -- (s2);
				\draw (f2) -- (s1);
				\draw (f2) -- (s2);
			\end{tikzpicture}\\
			\multicolumn{3}{c}{}\\
			\multicolumn{3}{c}{}\\
			\multicolumn{3}{c}{}\\
			Configuration D & Configuration E & Configuration F\\	
			\begin{tikzpicture}[thick,
				every node/.style={circle},
				fsnode/.style={fill=mygrey},
				ssnode/.style={fill=myblack},
				]

				\begin{scope}[xshift=0cm,yshift=-2cm,start chain=going below,node distance=15mm]
					\foreach \i in {1,2}
					\node[fsnode,on chain] (f\i) [label=left: \i] {};
				\end{scope}
				
				\begin{scope}[xshift=2cm,yshift=0cm,start chain=going below,node distance=15mm]
					\foreach \i in {1,2,3,4}
					\node[ssnode,on chain] (s\i) [label=right: \i] {};
				\end{scope}
				
				\draw (f1) -- (s1);
				\draw (f1) -- (s2);
				\draw (f2) -- (s3);
				\draw (f2) -- (s4);
			\end{tikzpicture}
			&\begin{tikzpicture}[thick,
				every node/.style={circle},
				fsnode/.style={fill=mygrey},
				ssnode/.style={fill=myblack},
				]

				\begin{scope}[xshift=0cm,yshift=-1cm,start chain=going below,node distance=15mm]
					\foreach \i in {1,2}
					\node[fsnode,on chain] (f\i) [label=left: \i] {};
				\end{scope}
				
				\begin{scope}[xshift=2cm,yshift=0cm,start chain=going below,node distance=15mm]
					\foreach \i in {1,2,3}
					\node[ssnode,on chain] (s\i) [label=right: \i] {};
				\end{scope}
				
				\draw (f1) -- (s1);
				\draw (f1) -- (s2);
				\draw (f2) -- (s2);
				\draw (f2) -- (s3);
			\end{tikzpicture}
			&\begin{tikzpicture}[thick,
				every node/.style={circle},
				fsnode/.style={fill=mygrey},
				ssnode/.style={fill=myblack},
				]

				\begin{scope}[xshift=0cm,yshift=0cm,start chain=going below,node distance=15mm]
					\foreach \i in {1,2,3}
					\node[fsnode,on chain] (f\i) [label=left: \i] {};
				\end{scope}
				
				\begin{scope}[xshift=2cm,yshift=0cm,start chain=going below,node distance=15mm]
					\foreach \i in {1,2,3}
					\node[ssnode,on chain] (s\i) [label=right: \i] {};
				\end{scope}
				
				\draw (f1) -- (s1);
				\draw (f1) -- (s2);
				\draw (f2) -- (s2);
				\draw (f2) -- (s3);
				\draw (f3) -- (s3);
				\draw (f3) -- (s1);
		\end{tikzpicture}\end{tabular}
	\end{center}
	{\footnotesize\textit{Notes: Worker nodes are in the left columns, indicated in grey. Firm nodes are in the right columns, indicated in black. In panels C and F we derive non-trivial moment restrictions on the parameter $\theta$ in model (\ref{eq_AKM_logit}), while no such restrictions exist in panels A,B,D,E.}}
\end{figure}

\paragraph{Configuration A: one worker in the same firm.}

Suppose worker $i$ stays in the same firm $j$ in both periods. We look for $\phi_{\theta}(y_{it},y_{i,t+1},x)$ such that, for $s\in\{0,1,2\}$,
\begin{align*}
	& \sum_{y_{it},y_{i,t+1}}\boldsymbol{1}\left\{y_{it}+y_{i,t+1}=s\right\}\phi_{\theta}(y_{it},y_{i,t+1},x)\exp\left((y_{it}x_{2it}+y_{i,t+1}x_{2i,t+1})'\theta\right)=0. \end{align*}
Given that $x_{2it}=x_{2i,t+1}$ (since the covariate does not vary within spell), this implies
\begin{align*}
	& \left(\sum_{y_{it},y_{i,t+1}}\boldsymbol{1}\left\{y_{it}+y_{i,t+1}=s\right\}\phi_{\theta}(y_{it},y_{i,t+1},x)\right)\exp\left(sx_{2it}'\theta\right)=0. \end{align*}
Hence $\theta$ drops out from the equation, and there is no information to estimate $\theta$ in this configuration.

\paragraph{Configuration B: one worker moving between two firms.}

Suppose worker $i$ moves between firms $j$ and $j'$. We look for $\phi_{\theta}(y_{it},y_{i,t+1},x)$ such that
\begin{align*}
	& \sum_{y_{it},y_{i,t+1}}\boldsymbol{1}\left\{y_{it}+y_{i,t+1}=s_1,y_{it}=s_2\right\}\phi_{\theta}(y_{it},y_{i,t+1},x)\exp\left((y_{it}x_{2it}+y_{i,t+1}x_{2i,t+1})'\theta\right)=0. \end{align*}
However, in this case,  $(s_1,s_2)$ fully determines $(y_{it},y_{i,t+1})$. Hence, for each $(s_1,s_2)$ we obtain $\phi_{\theta}(y_{it},y_{i,t+1},x)=0$, which shows there is no information about $\theta$ in this configuration.

\paragraph{Configuration C: two workers moving between the same two firms.}
Suppose workers $i$ and $i'$ both move between the same firms $j$ and $j'$. We look for a function $\phi_{\theta}(y_{it},y_{i,t+1},y_{i't},y_{i',t+1},x)$ such that \begin{align*}
	& \sum_{y_{it},y_{i,t+1},y_{i't},y_{i',t+1}}\boldsymbol{1}\left\{y_{it}+y_{i,t+1}=s_1,y_{i't}+y_{i',t+1}=s_2,y_{it}+y_{i't}=s_3,y_{i,t+1}+y_{i',t+1}=s_4\right\}\\
	&\quad \times\phi_{\theta}(y_{it},y_{i,t+1},y_{i't},y_{i',t+1},x)\exp\left((y_{it}x_{2it}+y_{i,t+1}x_{2i,t+1}+y_{i't}x_{2i't}+y_{i',t+1}x_{2i',t+1})'\theta\right)=0. \end{align*}

It turns out that, in this subnetwork configuration, there exist non-trivial moment restrictions on $\theta$. To see this, take $s_1=s_2=s_3=s_4=1$. We obtain
\begin{align*}
	& \phi_{\theta}(1,0,0,1,x)\exp\left((x_{2it}+x_{2i',t+1})'\theta\right)+\phi_{\theta}(0,1,1,0,x)\exp\left((x_{2i,t+1}+x_{2i't})'\theta\right)=0. \end{align*}
This implies the conditional moment restriction
\begin{align*}&\mathbb{E}\bigg[Y_{it}(1-Y_{i,t+1})(1-Y_{i't})Y_{i',t+1}\exp\left((X_{2i,t+1}+X_{2i't})'\theta\right)\\
	&\quad\quad \quad -(1-Y_{it})Y_{i,t+1}Y_{i't}(1-Y_{i',t+1})\exp\left((X_{2it}+X_{2i',t+1})'\theta\right)\,\big|\, X\bigg]=0.\end{align*}


\paragraph{Configuration D: two workers moving between different firms.}

Suppose worker $i$ moves between $j$ and $j'$, and worker $i'$ moves between different firms $j''$ and $j'''$. We look for $\phi_{\theta}(y_{it},y_{i,t+1},y_{i't},y_{i',t+1},x)$ such that
\begin{align*}
	& \sum_{y_{it},y_{i,t+1},y_{i't},y_{i',t+1}}\boldsymbol{1}\left\{y_{it}+y_{i,t+1}{=}s_1,y_{i't}+y_{i',t+1}{=}s_2,y_{it}{=}s_3,y_{i,t+1}{=}s_4,y_{i't}{=}s_5,y_{i',t+1}{=}s_6\right\}\\
	&\quad\quad \times\phi_{\theta}(y_{it},y_{i,t+1},y_{i't},y_{i',t+1},x)\exp\left((y_{it}x_{2it}+y_{i,t+1}x_{2i,t+1}+y_{i't}x_{2i't}+y_{i',t+1}x_{2i',t+1})'\theta\right)=0. \end{align*}
It is easy to see there is no non-trivial $\phi$ function in this case. Intuitively, since workers never share a firm, it is not possible to ``difference out'' the firm component of heterogeneity.

\paragraph{Configuration E: two workers moving to different firms from the same firm.}

Suppose worker $i$ moves between $j$ and $j'$, and worker $i'$ moves from the same firm $j$ to a different firm $j''$. We look for $\phi_{\theta}(y_{it},y_{i,t+1},y_{i't},y_{i',t+1},x)$ such that
\begin{align*}
	& \sum_{y_{it},y_{i,t+1},y_{i't},y_{i',t+1}}\boldsymbol{1}\left\{y_{it}+y_{i,t+1}{=}s_1,y_{i't}+y_{i',t+1}{=}s_2,y_{it}+y_{i't}{=}s_3,y_{i,t+1}{=}s_4,y_{i',t+1}{=}s_5\right\}\\
	&\quad\quad \times\phi_{\theta}(y_{it},y_{i,t+1},y_{i't},y_{i',t+1},x)\exp\left((y_{it}x_{2it}+y_{i,t+1}x_{2i,t+1}+y_{i't}x_{2i't}+y_{i',t+1}x_{2i',t+1})'\theta\right)=0. \end{align*}
It is easy to see there is no information about $\theta$ in this configuration.

\paragraph{Configuration F: three workers in a loop.}

There are many other subnetwork configurations providing information beyond configuration C. Indeed, consider three workers who move as follows: $i$ moves between firms $j$ and $j'$, $i'$ moves between $j'$ and $j''$, and $i''$ moves between $j''$ and $j$. We look for $\phi_{\theta}(y_{it},y_{i,t+1},y_{i't},y_{i',t+1},y_{i''t},y_{i'',t+1},x)$ such that
\begin{align*}
	& \sum_{y_{it},y_{i,t+1},y_{i't},y_{i',t+1},y_{i''t},y_{i'',t+1}}\boldsymbol{1}\bigg\{y_{it}+y_{i,t+1}=s_1,\\
	&\quad y_{i't}+y_{i',t+1}=s_2,y_{i''t}+y_{i'',t+1}=s_3,y_{it}+y_{i'',t+1}=s_4,y_{i't}+y_{i,t+1}=s_5,y_{i''t}+y_{i',t+1}=s_6\bigg\}\\
	&\quad \times\phi_{\theta}(y_{it},y_{i,t+1},y_{i't},y_{i',t+1},y_{i''t},y_{i'',t+1},x)\\
	&\quad\quad \times\exp\left((y_{it}x_{2it}+y_{i,t+1}x_{2i,t+1}+y_{i't}x_{2i't}+y_{i',t+1}x_{2i',t+1}+y_{i''t}x_{2i''t}+y_{i'',t+1}x_{2i'',t+1})'\theta\right)=0. \end{align*}
Taking $s_1=s_2=s_3=s_4=s_5=s_6=1$, one obtains
\begin{align*}
	& \phi_{\theta}(1,0,1,0,1,0,x)\exp\left((x_{2it}+x_{2i't}+x_{2i''t})'\theta\right)\\
	&\quad\quad \quad + \phi_{\theta}(0,1,0,1,0,1,x)\exp\left((x_{2i,t+1}+x_{2i',t+1}+x_{2i'',t+1})'\theta\right)=0.\end{align*}
This implies the conditional moment restriction
\begin{align*}&\mathbb{E}\bigg[Y_{it}(1-Y_{i,t+1})Y_{i't}(1-Y_{i',t+1})Y_{i''t}(1-Y_{i'',t+1})\exp\left((X_{2i,t+1}+X_{2i',t+1}+X_{2i'',t+1})'\theta\right)\\
	&\quad -(1-Y_{it})Y_{i,t+1}(1-Y_{i't})Y_{i',t+1}(1-Y_{i''t})Y_{i'',t+1}\exp\left((X_{2it}+X_{2i't}+X_{2i''t})'\theta\right)\,\big|\, X\bigg]=0.\end{align*}

\section{Average effects in logit network models\label{sec_logit_AME}}

 In this section we again consider model (\ref{eq_mod_logit}), and we study average effects of the form $$\mu=\mathbb{E}[m_{\theta}(A,X)],$$ for some known function $m_{\theta}$. In this case, (\ref{eq_FD_marg}) can be equivalently written as 
  \begin{align}
 	\frac{	\sum_{y\in\{0,1\}^n}\psi_{\theta}(y,x)\prod_{i=1}^n\exp(y_ix_{i1}'a+y_ix_{i2}'\theta)}{\prod_{i=1}^n\left(\exp(x_{i1}'a+x_{i2}'\theta)+1\right)}=m_{\theta}(a,x_{1},x_{2}).\label{eq_marg_char_logit}
 \end{align}	
 This equation characterizes the set of available moment restrictions on $\mu$, i.e., the set of $\psi$ functions such that (\ref{eq_mom_mu}) holds.
 
As a simple example, consider the case where $T=2$ in the static panel logit model (\ref{dgp_static_logit}), with a binary covariate $X_{it}$, and consider 
\begin{align*}m_{\theta}(a,x)&=\Pr(Y_{i1}=1\,|\, X_{i1}=1,A=a,\theta)-\Pr(Y_{i1}=1\,|\, X_{i1}=0,A=a,\theta)\\
	&= \frac{\exp(\theta +a)}{1+\exp(\theta +a)}-\frac{\exp(a)}{1+\exp(a)},\end{align*}
so that $\mu$ is an average partial effect. We show in Appendix \ref{app_logit} that no function $\psi$ satisfies (\ref{eq_marg_char_logit}). Intuitively, this comes from the fact that the distribution of $A$ given $X_{i1}=X_{i2}$ (i.e., for ``stayers'') is unidentified. 

In contrast, as we also show in Appendix \ref{app_logit}, the average partial effect of ``movers'', corresponding to 
\begin{align*}m_{\theta}(a,x)&=\boldsymbol{1}\{x_1\neq x_2\}\left[ \frac{\exp(\theta +a)}{1+\exp(\theta +a)}-\frac{\exp(a)}{1+\exp(a)}\right],\end{align*}
admits a characterization as in (\ref{eq_marg_char_logit}), whenever $\psi$ satisfies
\begin{align}
	&\psi_{\theta}(y_1,y_2,x)=0,\text{ for all }(y_1,y_2),\text{ if }x_1=x_2,\label{eq_psi_logit_1}\\
	&\psi_{\theta}(1,0,x)\exp(\theta x_1)+\psi_{\theta}(0,1,x)\exp(\theta x_2)=\exp(\theta)-1,\text{ if }x_1\neq x_2.\label{eq_psi_logit_2}
\end{align}
A simple example satisfying those conditions is 
$$\psi_{\theta}(y_1,y_2,x)=(x_2-x_1)(y_2-y_1),$$
as pointed out (in a more general nonparametric model) by \citet{chernozhukov2013average}. However, (\ref{eq_psi_logit_1}) and (\ref{eq_psi_logit_2}) imply additional moment restrictions. For example, one can take
\begin{align*}
    \psi_{\theta}(y_1,y_2,x)=\boldsymbol{1}\left\{x_1\neq x_2\right\}\left[y_1(1-y_2)\exp(\theta(1-x_1))-(1-y_1)y_2\exp(-\theta x_2)\right],
\end{align*}
which provides an additional moment restriction on $\mu$ under the logit model's assumptions.

It appears difficult to obtain moment equality restrictions on average partial effects in logit models on networks outside of the panel data case. As an example, consider the subnetwork configuration C in Figure \ref{Fig_bipart2}. In this case we have seen in the previous section how to obtain moment restrictions on $\theta$. However, we show in Appendix \ref{app_logit} that no function $\psi$ satisfies (\ref{eq_marg_char_logit}) for the average partial effect corresponding to 
\begin{align*}m_{\theta}(a,x)&= \frac{\exp(\theta +a_1+a_3)}{1+\exp(\theta +a_1+a_3)}-\frac{\exp(a_1+a_3)}{1+\exp(a_1+a_3)},\end{align*}
 where, in this model, $a_1$ is worker $i$'s fixed effect and $a_3$ is firm $j$'s fixed effect.

In models where no functional differencing restrictions are available, one may still be able to construct bounds on the average effect of interest. In panel data settings, this strategy was pursued by \citet{chernozhukov2013average}, \citet{davezies2021identification}, and \citet{dobronyi2021identification}, among others. However, implementing bounds approaches often requires estimating conditional moments given $X$. When $X$ represents a network matrix, conditional moment estimation may be especially challenging. In a panel data setting, \citet{pakel2021bounds} propose a bounding strategy that avoids the curse of dimensionality associated with conditioning covariates. Extending their approach to network settings is an interesting question for future work.

Lastly, in this section we have focused on binary choice models. The situation may be more favorable, in the sense of there existing informative functions $\phi$ and $\psi$, in models with continuous outcomes such as the CES specification (\ref{eq_AKM_CES}).

\section{Remarks on estimation\label{sec_remarks}}

To close our discussion, we briefly outline some possibilities for estimation of parameters and average effects, without providing details. 

Given a moment function $\phi$ as in Proposition \ref{propo_FD}, and a realization $(y,x)$ from the joint distribution of $(Y,X)$, one can estimate $\theta$ based on
$$\widehat{\theta}=\underset{\theta}{\mbox{argmin}}\, \left\|\phi_{\theta}(y,x)\right\|,$$
for some norm $\|\cdot\|$. In some models, this approach will deliver familiar estimators. For example, in the linear model (\ref{eq_mod_lin}), an estimator of $\beta$ based on (\ref{eq_restr_beta}) is
the ``quasi-differencing'' estimator\begin{equation}\widehat{\beta}=\left[x_2'(I_n-x_1x_1^{\dagger})x_2\right]^{-1}x_2'(I_n-x_1x_1^{\dagger})y,\label{eq_beta}\end{equation}
and an estimator of $\sigma^2$ based on (\ref{eq_restr_sigma2}) is the ``degree-of-freedom-corrected'' estimator
\begin{equation}\widehat{\sigma}^2=\frac{(y-x_2\widehat{\beta})'[I_n-x_1x_1^{\dagger}](y-x_2\widehat{\beta})}{\mbox{Trace}(I_n-x_1x_1^{\dagger})}.\label{eq_sigma2}\end{equation}

When constructing a function $\phi$ using the entire data is impractical, one can construct a set of functions $\phi_{\theta}^{(k)}(y,x)$ that depend on $y$ and $x$ only though a subset of the data. An estimator of $\theta$ is then
$$\widehat{\theta}=\underset{\theta}{\mbox{argmin}}\, \left\|\sum_{k=1}^K\phi_{\theta}^{(k)}(y,x)\right\|.$$
In the logit network formation model (\ref{eq_graham}), taking $\phi$ as in (\ref{eq_homoph1}), (\ref{eq_homoph2}), and (\ref{eq_homoph3}) (alongside its permutations), leads to estimators in the spirit of the tetrad logit estimator of \citet{graham2017econometric}.  

When focusing on average effects, a possible estimation approach based on Proposition \ref{prop_FD_marg} consists in setting
$$\widehat{\mu}=\psi_{\widehat{\theta}}(y,x),$$
for some estimator $\widehat{\theta}$. For example, in the linear model (\ref{eq_mod_lin}), an estimator of the quadratic form $\mu=\mathbb{E}[A'QA]$ based on (\ref{eq_psi_lin}) is \begin{equation*}\widehat{\mu}=(y-x_2\widehat{\beta})'(x_1^{\dagger})'Qx_1^{\dagger}(y-x_2\widehat{\beta})-\widehat{\sigma}^2\mbox{Trace}((x_1^{\dagger})'Qx_1^{\dagger}),\end{equation*}
where $\widehat{\beta}$ and $\widehat{\sigma}^2$ are given by (\ref{eq_beta}) and (\ref{eq_sigma2}), respectively. This corresponds to the bias-corrected estimator of \citet{andrews2008high}. For other average effects, regularization is typically needed for reliable estimation.

For all these estimators, there are important questions that remain to be addressed. What are their asymptotic properties (under suitable assumptions on how the network grows with the sample size)? How to conduct feasible inference on the population parameters? And, out of the available functions $\phi$ and $\psi$, how to choose a small subset of those (for tractability) without sacrificing too much precision (for efficiency)? Answering these questions will be an important task for future work. 

\clearpage

  \bibliographystyle{econometrica}
  \bibliography{biblio} 

\clearpage

\appendix

\begin{center}
	{ {\LARGE APPENDIX} }
\end{center}

\section{Proofs\label{app_proofs}}

\subsection{Proofs of Propositions \ref{propo_FD} and \ref{prop_FD_marg}}

Propositions \ref{propo_FD} and \ref{prop_FD_marg} follow directly from the following elementary lemma whose proof we include for completeness.

\begin{lemma}\label{lem_for_proofs}
	Let $g:{\cal{Z}}\rightarrow \mathbb{R}^p$ be a function such that ${\sup}_{z\in{\cal{Z}}}\, \|g(z)\|_{\infty}<\infty $, where $\|\cdot\|_{\infty}$ denotes the sup norm on $\mathbb{R}^p$. Suppose that, for all non-negative functions $h:{\cal{Z}}\rightarrow \mathbb{R}^{+}$ such that $\int_{\cal{Z}}h(z)dz=1$, we have $\int_{\cal{Z}}g(z)h(z)dz=0$. Then $g(z)=0$ almost everywhere on ${\cal{Z}}$.
	
\end{lemma}
\begin{proof}
	Let $k\in\{1,...,p\}$, and, for all $z\in{\cal{Z}}$, let $g_k(z)$ denote the $k$-th element of $g(z)$. For all $\ell\in {\cal{L}}^{1}(\cal{Z})$, let $\|\ell\|_1=\int_{\cal{Z}}|\ell(z)|dz$, $\ell^{+}(z)=\mbox{max}(\ell(z),0)$, and $\ell^{-}(z)=-\mbox{min}(\ell(z),0)$. For all $\ell\in {\cal{L}}^{1}(\cal{Z})$ we have
	\begin{align*}\int_{\cal{Z}}g_k(z)\ell(z)dz=\|\ell^{+}\|_1\int_{\cal{Z}}g_k(z)\frac{\ell^{+}(z)}{\|\ell^{+}\|_1}dz-\|\ell^{-}\|_1\int_{\cal{Z}}g_k(z)\frac{\ell^{-}(z)}{\|\ell^{-}\|_1}dz=0,\end{align*}
	where we have used that $\frac{\ell^{+}(z)}{\|\ell^{+}\|_1}$ and $\frac{\ell^{-}(z)}{\|\ell^{-}\|_1}$ are non-negative and integrate to one (with the convention $0/0=0$ whenever $\|\ell^{+}\|_1=0$ or $\|\ell^{-}\|_1=0$). Since $g_k$ is bounded, it follows that 
	$${\sup}_{z\in{\cal{Z}}}\, |g_k(z)|={\sup}_{\ell\in {\cal{L}}^{1}({\cal{Z}})\, :\, \|\ell\|_1=1}\,\left|\int_{\cal{Z}}g_k(z)\ell(z)dz\right| =0,$$
	so $g_k=0$ almost everywhere on ${\cal{Z}}$. Lastly, since this holds for all $k\in\{1,...,p\}$, it follows that $g=0$ almost everywhere on ${\cal{Z}}$.
\end{proof}

\paragraph{Proof of Proposition \ref{propo_FD}.} It is sufficient to show that (i) implies (ii). Suppose that (i) holds. Let $Z=(A,X)$, and $g(Z)=\mathbb{E}\left[\phi_{\theta}(Y,X)\,|\, A,X\right]$. It follows from (i) and Lemma \ref{lem_for_proofs} that $g=0$ almost everywhere on ${\cal{Z}}$. This shows (ii) and completes the proof.

\paragraph{Proof of Proposition \ref{prop_FD_marg}.} Let $Z=(A,X)$, and $g(Z)=\mathbb{E}\left[\psi_{\theta}(Y,X)\,|\, A,X\right]-m_{\theta}(A,X)$. It follows from (i) and Lemma \ref{lem_for_proofs} that $g=0$ almost everywhere on ${\cal{Z}}$. This shows (ii) and completes the proof.

\subsection{Proof of Proposition \ref{prop_linear_quad}}

Let $n_1=n-n_2$. Then $u_1$ is an $n\times n_1$ matrix, and $v_1$ is an $n_1\times 1$ vector.\\
Let
$$\varphi_{\theta}(v_1,x)= \frac{1}{(2\pi \sigma^2)^{\frac{n_2}{2}}}\int \psi_{\theta}(x_2\beta+u_1v_1+u_2v_2,x)\exp\left(-\frac{1}{2\sigma^2}v_2'v_2\right)dv_2.$$
It follows from (\ref{eq_prop4_forproof}) that (\ref{eq_FD_marg}) is equivalent to
\begin{align*}
	\frac{1}{(2\pi \sigma^2)^{\frac{n_1}{2}}}\int \varphi_{\theta}(v_1,x)\exp\left(-\frac{1}{2\sigma^2}(v_1-u_1'x_1a)'(v_1-u_1'x_1a)\right)dv_1=a'Qa.	\end{align*}
As a result, (\ref{eq_FD_marg}) is equivalent to
\begin{align*}
	& \frac{1}{(2\pi \sigma^2)^{\frac{n_1}{2}}}\int \varphi_{\theta}(v_1,x)\exp\left(-\frac{1}{2\sigma^2}(v_1-b)'(v_1-b)\right)dv_1=b'((u_1'x_1)^{\dagger})'Q(u_1'x_1)^{\dagger}b,	\end{align*}
where $b=u_1'x_1 a$, and we have used that $x_1$ has full column rank.\\
Let ${\cal{F}}$ denote the Fourier transform operator. For any integrable function $f:\mathbb{R}^{n_1}\rightarrow \mathbb{R}$ we have, for all $s\in \mathbb{R}^{n_1}$,
$${\cal{F}}[f](s)=\int f(x)\exp(\boldsymbol{i}s'x)dx,$$
where $\boldsymbol{i}$ is a complex root of $-1$, and the integral is over $\mathbb{R}^{n_1}$.\\
We have
$${\cal{F}}\left[\varphi_{\theta}(v_1,x)\right](s)\exp\left(-\frac{\sigma^2}{2}s's\right)={\cal{F}}\left[b'((u_1'x_1)^{\dagger})'Q(u_1'x_1)^{\dagger}b\right](s).$$
Let $C=((u_1'x_1)^{\dagger})'Q(u_1'x_1)^{\dagger}$. We have, for $\delta(\cdot)$ the Dirac delta function,
\begin{align*}{\cal{F}}\left[b'Cb\right](s)&=\sum_{i,j}c_{ij}{\cal{F}}\left[b_ib_j\right](s)\\
	&=\sum_ic_{ii} {\cal{F}}\left[b_i^2\right](s)+\sum_{i\neq j}c_{ij}{\cal{F}}\left[b_ib_j\right](s)\\
	&=\sum_ic_{ii} {\cal{F}}\left[b_i^2\right](s)+\sum_{i\neq j}c_{ij}{\cal{F}}\left[b_i\right](s){\cal{F}}\left[b_j\right](s)\\
	&=\sum_ic_{ii} [-(2\pi)\delta''(s_i)](2\pi)^{n_1-1}\prod_{j\neq i}\delta(s_j)\\
	&\quad +\sum_{i\neq j}c_{ij}[-i(2\pi)\delta'(s_i)][-i(2\pi)\delta'(s_j)](2\pi)^{n_1-2}\prod_{k\neq(i,j)}\delta(s_k)\\
	&=-(2\pi)^{n_1}\left(\sum_ic_{ii} \delta''(s_i)\prod_{j\neq i}\delta(s_j)+\sum_{i\neq j}c_{ij}\delta'(s_i)\delta'(s_j)\prod_{k\neq(i,j)}\delta(s_k)\right).
\end{align*}
Since this holds for all $s$, the Fourier inversion theorem gives
\begin{align*}\varphi_{\theta}(v_1,x)&=\frac{1}{(2\pi)^{n_1}}\int {\cal{F}}\left[b'Cb\right](s)\exp\left(\frac{\sigma^2}{2}s's\right) \exp\left(-iv_1's\right)ds\\
	&=-\sum_ic_{ii}\int \delta''(s_i)\exp\left(\frac{\sigma^2}{2}s_i^2\right) \exp\left(-iv_{1i}s_i\right)ds_i\\
	&\quad -\sum_{i\neq j}c_{ij}\iint \delta'(s_i)\delta'(s_j)\exp\left(\frac{\sigma^2}{2}(s_i^2+s_j^2)\right) \exp\left(-i(v_{1i}s_i+v_{1j}s_j)\right)ds_ids_j.\end{align*}
Now, by integration by parts, we have
\begin{align*}&\int \delta'(s_i)\exp\left(\frac{\sigma^2}{2}s_i^2\right) \exp\left(-iv_{1i}s_i\right)ds_i=iv_{1i},\\
	&	\int \delta''(s_i)\exp\left(\frac{\sigma^2}{2}s_i^2\right) \exp\left(-iv_{1i}s_i\right)ds_i=\sigma^2-v_{1i}^2.
\end{align*}
Hence,
\begin{align*}\varphi_{\theta}(v_1,x)&=-\sum_ic_{ii}(\sigma^2-v_{1i}^2)+\sum_{i\neq j}c_{ij}v_{1i}v_{1j}\\
	&=v_1'Cv_1-\sigma^2\mbox{Trace}(C),\end{align*}
and
\begin{align*} &\frac{1}{(2\pi \sigma^2)^{\frac{n_2}{2}}}\int \psi_{\theta}(x_2\beta+u_1v_1+u_2v_2,x)\exp\left(-\frac{1}{2\sigma^2}v_2'v_2\right)dv_2\\&=v'u_1((u_1'x_1)^{\dagger})'Q(u_1'x_1)^{\dagger}u_1'v-\sigma^2\mbox{Trace}((x_1^{\dagger})'Qx_1^{\dagger})\\
	&=v'(x_1^{\dagger})'Qx_1^{\dagger}v-\sigma^2\mbox{Trace}((x_1^{\dagger})'Qx_1^{\dagger}).\end{align*}
This concludes the proof of Proposition \ref{prop_linear_quad}.
\section{Other average effects in the linear model\label{app_linear_nonquad}}

Following the arguments in the proof of Proposition \ref{prop_linear_quad}, we have
$${\cal{F}}\left[\varphi_{\theta}(v_1,x)\right](s)\exp\left(-\frac{\sigma^2}{2}s's\right)={\cal{F}}\left[m_{\theta}\left((u_1'x_1)^{\dagger}b,x\right)\right](s).$$
Hence, provided the Fourier inversion theorem can be applied, we have 
$$\varphi_{\theta}(v_1,x)={\cal{F}}^{-1}\left[{\cal{F}}\left[m_{\theta}\left((u_1'x_1)^{\dagger}b,x\right)\right](s)\exp\left(\frac{\sigma^2}{2}s's\right)\right](v_1),$$
and thus
\begin{align*}
	&\frac{1}{(2\pi \sigma^2)^{\frac{n_2}{2}}}\int \psi_{\theta}(x_2\beta+u_1v_1+u_2v_2,x)\exp\left(-\frac{1}{2\sigma^2}v_2'v_2\right)dv_2\\&\quad\quad \quad={\cal{F}}^{-1}\left[{\cal{F}}\left[m_{\theta}\left((u_1'x_1)^{\dagger}b,x\right)\right](s)\exp\left(\frac{\sigma^2}{2}s's\right)\right](v_1).
\end{align*}
As a special case, suppose $\psi_{\theta}(y,x)$ is a function of $u_1'(y-x_2\beta)$ and $x$ only. Then
\begin{align*}
	&\psi_{\theta}(y,x)={\cal{F}}^{-1}\left[{\cal{F}}\left[m_{\theta}\left((u_1'x_1)^{\dagger}b,x\right)\right](s)\exp\left(\frac{\sigma^2}{2}s's\right)\right](u_1'(y-x_2\beta)).
\end{align*}

\section{Average effects in logit models\label{app_logit}}

Let
\begin{align*}m_{\theta}(a,x)&=\frac{\exp(\theta +a)}{1+\exp(\theta +a)}-\frac{\exp(a)}{1+\exp(a)}.\end{align*}
In the panel data case with $T=2$, (\ref{eq_marg_char_logit}) can equivalently be written as
\begin{align*}&\frac{\exp(\theta +a)}{1+\exp(\theta +a)}-\frac{\exp(a)}{1+\exp(a)}\\&=\frac{	\sum_{(y_{i1},y_{i2})\in\{0,1\}^2}\psi_{\theta}(y_{i1},y_{i2},x_{i})\exp((y_{i1}+y_{i2})a+(y_{i1}x_{i1}+y_{i2}x_{i2})\theta)}{\left(\exp(a+x_{i1}\theta)+1\right)\left(\exp(a+x_{i2}\theta)+1\right)}.\end{align*}

That is,
\begin{align}&(\exp(\theta)-1)\exp(a)\left(\exp(a+x_{i1}\theta)+1\right)\left(\exp(a+x_{i2}\theta)+1\right)\notag\\&=(1+\exp(\theta +a))(1+\exp(a))\notag\\
	&\quad\times	\sum_{(y_{i1},y_{i2})\in\{0,1\}^2}\psi_{\theta}(y_{i1},y_{i2},x_{i})\exp((y_{i1}+y_{i2})a+(y_{i1}x_{i1}+y_{i2}x_{i2})\theta).\label{eq_inter}\end{align}

Let $Z=\exp(a)$. The coefficient of $Z^0$ on the left-hand side of (\ref{eq_inter}) is equal to zero, and the coefficient on the right-hand side is $\psi_{\theta}(0,0,x_{i})$. It thus follows that $\psi_{\theta}(0,0,x_{i})=0$.

The coefficient of $Z^4$ on the left-hand side of (\ref{eq_inter}) is equal to zero, and the coefficient on the right-hand side is $\exp((x_{i1}+x_{i2}+1)\theta)\psi_{\theta}(1,1,x_{i})$. It thus follows that $\psi_{\theta}(1,1,x_{i})=0$.
  
Hence,
\begin{align}&(\exp(\theta)-1)\left(\exp(a+x_{i1}\theta)+1\right)\left(\exp(a+x_{i2}\theta)+1\right)\notag\\&=(1+\exp(\theta +a))(1+\exp(a))	\left[\psi_{\theta}(1,0,x_{i})\exp(x_{i1}\theta)+\psi_{\theta}(0,1,x_{i})\exp(x_{i2}\theta)\right].\label{eq_inter3}\end{align}
So the existence of valid moment functions requires that the ratio $$\frac{(\exp(\theta)-1)\left(\exp(a+x_{i1}\theta)+1\right)\left(\exp(a+x_{i2}\theta)+1\right)}{(1+\exp(\theta +a))(1+\exp(a))}$$ does not depend on $a$. This is not possible if $x_{i1}=x_{i2}$. However, if $x_{i1}\neq x_{i2}$ then (\ref{eq_inter3}) simplifies to 
\begin{align*}&\exp(\theta)-1=	\psi_{\theta}(1,0,x_{i})\exp(x_{i1}\theta)+\psi_{\theta}(0,1,x_{i})\exp(x_{i2}\theta),\end{align*}
which is indeed identical to (\ref{eq_psi_logit_2}).

 Next, consider Configuration C in Figure \ref{Fig_bipart2}, and let 
 \begin{align*}m_{\theta}(a,x)&= \frac{\exp(\theta +a_1+a_3)}{1+\exp(\theta +a_1+a_3)}-\frac{\exp(a_1+a_3)}{1+\exp(a_1+a_3)}.\end{align*}
 
 By (\ref{eq_marg_char_logit}), we look for $\psi$ such that
   \begin{align*}
 &\frac{\exp(\theta +a_1+a_3)}{1+\exp(\theta +a_1+a_3)}-\frac{\exp(a_1+a_3)}{1+\exp(a_1+a_3)}=	\frac{	\sum_{y\in\{0,1\}^4}\psi_{\theta}(y,x)\prod_{i=1}^4\exp(y_ix_{i1}'a+y_ix_{i2}'\theta)}{\prod_{i=1}^4\left(\exp(x_{i1}'a+x_{i2}'\theta)+1\right)}.
 \end{align*}	
 
 That is,
 \begin{align}
 	&(\exp(\theta)-1)\exp(a_1+a_3)\left(\exp(a_1+a_3+x_{12}\theta)+1\right)\left(\exp(a_1+a_4+x_{22}\theta)+1\right)\notag\\
 	&\quad \times\left(\exp(a_2+a_3+x_{32}\theta)+1\right)\left(\exp(a_2+a_4+x_{42}\theta)+1\right)\notag\\
 	&=	\sum_{y\in\{0,1\}^4}\psi_{\theta}(y,x)\exp\bigg(y_1(a_1+a_3+x_{12}\theta)+y_2(a_1+a_4+x_{22}\theta)+y_3(a_2+a_3+x_{32}\theta)\notag\\
 	&\quad\quad\quad  +y_4(a_2+a_4+x_{42}\theta)\bigg)(1+\exp(\theta +a_1+a_3))(1+\exp(a_1+a_3)).\label{eq_inter2}
 \end{align}	

Let $Z_k=\exp(a_k)$ for $k\in\{1,...,4\}$. The left-hand side and right-hand side in (\ref{eq_inter2}) are polynomials in $(Z_1,...,Z_4)$. Since they are equal to each other for all $a_1,...,a_4$ in $\mathbb{R}$, hence for all $Z_1,...,Z_4$ in $\mathbb{R}_{>0}$, they are equal to each other for all $Z_1,...,Z_4$ in $\mathbb{R}$ as well. Suppose next that $\theta\neq 0$. The right-hand side in (\ref{eq_inter2}) is a multiple of $(1+Z_1Z_3)$ and $(1+\exp(\theta)Z_1Z_3)$. However, the left-hand side in (\ref{eq_inter2}) is either a multiple of $(1+Z_1Z_3)$ or a multiple of $(1+\exp(\theta)Z_1Z_3)$ (depending on the value of $x_{12}$) but it is not a multiple of both terms. It follows that, when $\theta\neq 0$, there is no $\psi$ satisfying (\ref{eq_inter2}). A direct comparison of the monomial terms of the polynomials on both sides of equation (\ref{eq_inter2}) confirms this argument.

\end{document}